\documentclass[11pt,a4paper]{article}
\usepackage{amsmath,amsthm,latexsym,amssymb,amsfonts,epsfig,parskip}

\newtheorem{lemm}{Lemma}[section]
\newtheorem{prop}[lemm]{Proposition}

\newtheorem{defi}[lemm]{Definition}

\newcommand{\scpr}[2]{\left\langle#1\,,\, #2 \right\rangle}

\newcommand{\betr}[1]{\lvert #1 \rvert}
\newcommand{\R}{\mathbb{R}}                  
\newcommand{\C}{\mathbb{C}}                  
\newcommand{\RB}{\R_{\text{B}}}   
\newcommand{\RBd}{\widehat{\R}_{\text{B}}}   
\newcommand{\Sch}{\mathcal{S}}                 
\newcommand{\Schd}{\mathcal{S}'}
\newcommand{\wh}[1]{\widehat{#1}}
\newcommand{\ftn}[1]{\widetilde{#1}}
\newcommand{\ft}[2]{\wh{#1}_{#2}}            
\newcommand{\ftcc}[2]{\overline{\wh{#1}}_{#2}}    
\newcommand{\ch}[2]{h_{#1}(#2)}            

\newcommand{\op}[1]{\boldsymbol{#1}}

\DeclareMathOperator{\symb}{Symb}
\DeclareMathOperator{\cyl}{Cyl(\RB)}              
\DeclareMathOperator{\cyld}{Cyl(\RBd)}
\DeclareMathOperator{\cylcyld}{Cyl(\RB\times\RBd)}
\DeclareMathOperator{\cyldd}{\cyld\!{}^*}
\DeclareMathOperator{\re}{\text{Re}}
\DeclareMathOperator{\im}{\text{Im}}
\DeclareMathOperator{\sign}{\text{Sign}}

\title{Phase space quantization and Loop Quantum Cosmology: A Wigner function
for the Bohr-compactified real line}
\author{Christopher J. Fewster$^1$ and Hanno Sahlmann$^2$}

\date{1. Department of Mathematics,
University of York,\\ Heslington, York YO10 5DD, U.K.\\
2. Spinoza Institute, Universiteit Utrecht\\[.3cm]\today\\[.3cm]
{\small Preprint ITP-UU-08/10, SPIN-08/10}}
\begin{document}
\maketitle
\begin{abstract}
We give a definition for the Wigner function for
quantum mechanics on the Bohr compactification of
the real line and prove a number of simple
consequences of this definition. We then discuss
how this formalism can be applied to loop quantum
cosmology. As an example, we use the Wigner
function to give a new quantization of an
important building block of the Hamiltonian
constraint.
\end{abstract}
\section{Introduction}
The Wigner function \cite{wigner} has long been
recognized as a tool in quantum mechanics. For a
wave function $\Psi(x)$ on the real line it is defined as
\begin{equation*}
 W(\Psi)(x,p)\doteq\int
 \overline{\Psi}\left(x+\frac{1}{2}x'\right)\Psi\left(x-\frac{1}{2}x'\right)e^{ipx'}\,\text{d}x'.
\end{equation*}
It is a function on phase space that comes, in a
certain sense, as close to being a classical
probability distribution corresponding to $\Psi$
on phase space as possible. It can therefore be
used to analyze the extent to which a given quantum
state can be described in classical terms.
Furthermore the Wigner function
figures prominently in Weyl quantization, a
map that assigns symmetric operators to real functions
(subject to smoothness and fall-off criteria)
on phase space in a systematic fashion. In physics parlance,
Weyl quantization is referred to as \emph{totally symmetric ordering}.
A comprehensive mathematical treatment of the Wigner function and its properties
in quantum mechanics can be found in
\cite{folland}.

Loop quantum cosmology (LQC for short, see \cite{Bojowald:2006da} for
a review) is a theory of quantum cosmology developed in close
connection with loop quantum gravity
\cite{Ashtekar:2004eh,Ashtekar:2007tv,Thiemann:2001yy}, and can be
viewed as a symmetry reduced version of the latter. It has been used
as a testbed for techniques used in loop quantum gravity, but it can
be argued that it also makes physical predictions in its own right.
One large set of results shows that the classical singularities of
cosmology are resolved in the quantum theory. As an example, we refer
to \cite{Ashtekar:2006uz} for a beautiful result in this direction.

Technically LQC started out as a quantum theory on a circle, but it
was later realized that it is actually more appropriately formulated
as a quantum theory on the Bohr compactification $\RB$ of the real
line. While we will review some of the mathematics of $\RB$ and of
the functions on this space in section \ref{se_rbohr} below, we refer
to \cite{Ashtekar:2003hd} for a good overview over both mathematical
and physical aspects of these developments.

Given that quantum mechanics on $\RB$ is the foundation
for LQC, it is an interesting question whether the
Wigner function can be generalized to this setting.
A look at the literature
shows that the Wigner function can be, and has been, generalized in a
number of ways, for example to quantum
mechanics on U(1), and more generally, certain non-Abelian
groups (see for example \cite{ali,mukunda}).
In fact, the Wigner function for
U(1) has made a brief appearance in LQG
\cite{Bojowald:2001ep}, where it was used in the
study of the semiclassical limit. In these
generalizations, the role of Fourier analysis is
played by its natural generalizations for
harmonic analysis on groups (Pontryagin duality
in the Abelian case, Peter-Weyl theory for
compact non-Abelian groups). It must however be
said that such generalizations are generically
neither unique, nor do they share all the
properties of the Wigner function on $\R$.

To the best of our knowledge however, a
generalization of the Wigner function to $\RB$
has not yet been considered. The present paper
intends to fill this gap. As it turns out, the
generalization of the Wigner function to $\RB$ is
quite straightforward. What is more, its properties mirror that of
its cousin on $\R$ extremely closely. This is due
on the one hand to the Abelian nature of $\RB$,
on the other hand to a useful property of
the Pontryagin dual of $\RB$, namely that the
operation of ``taking a square root" with respect
to its group product is well defined. Such square
roots (or divisions by two, in additive notation)
will naturally show up when proving properties of
the Wigner function.

We will also demonstrate the applicability of the
Wigner function to issues in LQC. In particular
we will use it to obtain the Weyl quantization of
the modified holonomy from
\cite{Ashtekar:2006wn}, which in turn could be
used to define a modified quantum dynamics. We
will compare the properties of this quantization
to the standard one, but we will not yet use it
to complete the quantization of the Hamiltonian
constraint and attempt an analysis of the
physical differences that would result.

We should say that there are other conceivable
applications of the Wigner functions besides the
one we demonstrate in this paper. To give an
example we recall that recently a method has been
established in LQC that allows to calculate
effective equations of motion within a systematic
approximation scheme \cite{Bojowald:2005cw,
Bojowald:2006ww}. Weyl ordering figures
prominently in this method and thus we expect
that the Wigner function techniques from the
present paper may also be useful in that context.

The paper is organized as follows: We start by giving a brief review of the
properties of the Wigner function in ordinary
quantum mechanics in section \ref{se_review}. In
section \ref{se_rbohr} we generalize its
definition to quantum mechanics on the Bohr
compactification $\RB$, and list analogous
properties. In section \ref{se_applications} we
sketch an application to LQC. We finish with a
discussion of our results and the possibility of
a generalization to loop quantum gravity in
section \ref{se_discussion}.
\section{The Wigner function on $\R$}
\label{se_review}
In the present section we will recall the
definition of the Wigner function for quantum
mechanics on the real line. We will follow
closely the exposition in \cite{folland} (and refer to it for proofs and
details) although with slightly
different conventions. Our Fourier transform convention will be
\begin{equation*}
\ftn{f}(k)\doteq \int f(x) e^{-ikx}\, dx,
\qquad f(x)=\int \ftn{f}(k) e^{ikx}\, \frac{dk}{2\pi};
\end{equation*}
we write the scalar product on
$L^2(\R,dx)$ as $\scpr{\cdot}{\cdot}$, and denote the
Schwartz test functions on $\R$ by $\Sch(\R)$.
We
introduce the usual position and momentum operators $\op{X}$ and $\op{P}$
\begin{equation}
\label{eq_schroedinger}
\op{X}\Psi(x)=x\Psi(x)\qquad\text{and}\qquad
\op{P}\Psi(x)=\frac{1}{i} \frac{d}{dx} \Psi(x),
\end{equation}
for $\Psi\in\Sch(\R)$.

The Wigner function $W(\Psi)$ of a wave function is conveniently defined
as a special case $W(\Psi)=W(\Psi,\Psi)$ of the Wigner transform $W(\Psi,\Psi')$ of a \emph{pair}
of wave functions, in turn defined by the following equivalent expressions:
\begin{align}
\label{eq_wig1}
 W(\Psi,\Psi')(x,p)&\doteq\int
 \overline{\Psi}(x+x'/2)\Psi'(x-x'/2)e^{ipx'}\,\text{d}x'\\
&=2\int\overline{\ftn{\Psi}}(k)\ftn{\Psi}'(2p-k)e^{2
ix(p-k)}\,\frac{\text{d}k}{2\pi}\label{eq_wig3}\\
&=\int\overline{\ftn{\Psi}}(p-q/2)\ftn{\Psi}'(p+q/2)e^{ixq}\,\frac{\text{d}q}{2\pi}\label{eq_wig4}
\end{align}
Some of its basic properties are as follows (see Props. 1.92 and 1.96
in~\cite{folland}):
\begin{prop}
Eq.\eqref{eq_wig1} defines the Wigner transform as a sesquilinear map
between the following spaces:
\begin{equation*}
W:\begin{array}{ccl}
\Sch (\R) \times \Sch (\R) &\longrightarrow &\Sch (\R^2)\\
L^2 (\R) \times L^2 (\R) &\longrightarrow &L^2(\R^2)\cap C_\infty(\R^2)\\
\Schd (\R) \times \Schd (\R) &\longrightarrow &\Schd (\R^2),
\end{array}
\end{equation*}
with the last of these a continuous extension of the others, and
$C_\infty(\R^2)$ denoting continuous functions vanishing at infinity. 
On $L^2(\R)\times L^2(\R)$ the Wigner transform has the overlap property
\begin{equation}
\iint \overline{W(\Psi_1,\Psi_2)}
W(\Phi_1,\Phi_2)\, \frac{\textrm{d}x \textrm{d}p}{2\pi}
=\overline{\scpr{\Psi_1}{\Phi_1}}\scpr{\Psi_2}{\Phi_2}
\label{eq_Wig_ip}
\end{equation}
and is hermitian:
\begin{equation}
\label{eq_hermite}
\overline{W(\Psi,\Psi')}=W(\Psi',\Psi).
\end{equation}
Furthermore, for $\Psi\in\Sch(\R)$, $W(\Psi,\Psi)$ has marginal distributions
\begin{equation}
\label{eq_marginal} \int W(\Psi,\Psi)(x,k)\,
\textrm{d}x =\betr{\ftn{\Psi}}^2(k),\qquad \int
W(\Psi,\Psi)(x,k)\, \frac{\textrm{d}k}{2\pi} =\betr{\Psi}^2(x)
\end{equation}
and the same is true (modulo technical details) for $\Psi\in L^2(\R)$.
\end{prop}
We note in particular that according to \eqref{eq_hermite},
$W(\Psi,\Psi)$ is real. Moreover, its
``marginals'' are the quantum mechanical
probability distributions for measurements of
position and momentum according to
\eqref{eq_marginal}. However, $W(\Psi,\Psi)$ fails to be a
joint probability distribution because
it is not positive in general. It is
positive precisely for Gaussian states, arguably the most
classical quantum states:
\begin{prop}[Hudson's theorem \cite{hudson}]
\label{pr_gauss} For $\Psi\in L^2(\R)$, $\Psi\neq
0$ it holds that
\begin{equation*}
W(\Psi,\Psi)\geq 0 \qquad \Longleftrightarrow \qquad \Psi \text{ is a
Gaussian}.
\end{equation*}
\end{prop}
In this sense the Wigner function is as close to
a classical joint probability distribution in
position and momentum as quantum mechanics
allows.

The Wigner function can also be used in
quantization: As per \eqref{eq_marginal},
\begin{equation*}
\begin{split}
\iint x W(\Psi,\Psi')(x,p) \,\frac{\text{d}x\text{d}p}{2\pi} =
\scpr{\Psi}{\op{X}\Psi'},\\
\iint p W(\Psi,\Psi')(x,p) \,\frac{\text{d}x\text{d}p}{2\pi}  =
\scpr{\Psi}{\op{P}\Psi'}
\end{split}
\end{equation*}
whenever $\Psi'$ is in the domain of the
respective operators. Also it is easy to
calculate
\begin{equation*}
\iint xp W(\Psi,\Psi')(x,p) \,\frac{\text{d}x\text{d}p}{2\pi}=
\scpr{\Psi}{\frac{1}{2}(\op{PX}+\op{XP})\Psi'}
\end{equation*}
which shows a relation to quantization by
symmetric ordering. In fact this relation carries on much
further: For phase space functions $\sigma(x,p)$ in $L^1(\R^2)$,
we may define a quantization $\op{\sigma}$ of $\sigma$ by
the Bochner integral
\begin{equation*}
\op{\sigma}\doteq \int \ftn{\sigma}(p,x)
\exp[i(p\op{X}+x\op{P})]\, \frac{\text{d}x\text{d}p}{2\pi}
\end{equation*}
where
\begin{equation*}
\ftn{\sigma}(p,x) \doteq \int
\sigma(x',p')e^{-i(px'+xp')}\frac{\text{d}x'\text{d}p'}{2\pi}.
\end{equation*}
One can even extend this definition to $\sigma(x,p)$
which are in $\Schd(R^2)$. The matrix elements of this operator are
determined in a very simple manner from the
Wigner function and the symbol $\sigma$:
\begin{prop}[Prop.\ (2.5) in \cite{folland}]
\label{pr_calculus}
For $\sigma\in \Schd(R^2)$ and $\Psi,\Psi'\in \Sch(\R)$,
\begin{equation*}
\scpr{\Psi}{\op{\sigma}\Psi'}=\iint\sigma(x,p)W(\Psi,\Psi')(x,p)\,
\frac{\textrm{d}x\textrm{d}p}{2\pi}.
\end{equation*}
\end{prop}
After this brief exposition of the properties of the Wigner function on
$\R$, we can turn now to the actual topic of this paper.
\section{The Wigner function on $\RB$}
\label{se_rbohr}
\subsection{The Bohr compactification $\RB$}
In the present section we define the
Wigner function for wave functions on the Bohr
compactification of the real line, and derive
some of its properties. We will start by
recalling some basic facts about harmonic
analysis on Abelian groups and the definition of
$\RB$. A good reference for these matters is \cite{reiter}.

Given any locally compact Abelian group $G$, one
can form the dual group $\widehat{G}$ as the
Abelian group of (continuous) characters of $G$.
Multiplication in $\widehat{G}$ is given by pointwise
multiplication of characters, the inverse by
complex conjugation, and the topology by uniform
convergence on compact sets. With this topology
$\widehat{G}$ itself becomes a locally compact
group. There is a natural isomorphism between $G$ and its double-dual.

As locally compact Abelian groups, $G$ and
$\widehat{G}$  have unique (up to scaling) Haar
measures $d\mu$, $d\widehat{\mu}$. Fourier
transform can be defined as
\begin{equation*}
\widehat{f}(\chi)=\int_G d\mu\,
f(x)\overline{\chi}(x)
\end{equation*}
for a character $\chi$ of $G$.
The normalization of the Haar measures can be
chosen such that Fourier transform becomes an isomorphism
\begin{equation*}
L^2(G,d\mu)\longrightarrow
L^2(\widehat{G},d\widehat{\mu}).
\end{equation*}
A locally compact Abelian group $G$ is compact
iff $\widehat{G}$ is discrete. This is used to
define the Bohr compactification $B(G)$ of a
locally compact group $G$: $B(G)$ is defined as
the dual group of $\widehat{G}_\text{discr}$, with the
latter being $\widehat{G}$ as far as group structure is
concerned, but equipped with the
\textit{discrete} topology. For the reals this
works out as follows.

Let $G$ be the additive group of real numbers $G=(\R,+)$ with its usual
topology, which we think of as the configuration space of ordinary
quantum mechanics in one dimension. The characters of $G$ are precisely the functions
$h_\mu:G\to\C$ given by
\begin{equation}\label{eq:char_form}
\ch{\mu}{c} = \exp[i\mu c]
\end{equation}
labelled by $\mu\in\R$, and form a group $\widehat{G}$ isomorphic to
$(\R,+)$ with the usual topology, i.e., the usual momentum space. Thus
$\widehat{G}_\text{discr}$ is the additive group of real numbers with
the discrete topology. Since this group is discrete and Abelian, its
characters form a compact Abelian group $B(\R)$, which we will also
denote $\RB$. Now each real number
$c$ defines an obvious character $\RBd\owns\mu\mapsto\ch{\mu}{c}\in\C$ of
$\RBd$ and this correspondence embeds $(\R,+)$ as a dense subgroup
of $\RB$, justifying the description of $\RB$ as a compactification of the real line.
The reason that there are more characters of $\RBd$ than of
$\widehat{G}$ is that there is now no continuity requirement in the $\mu$
variable (or more precisely, continuity is required with respect to the
discrete topology). It will be convenient to denote the character on $\RBd$ corresponding to
any $c\in\RB$ by $\mu\mapsto\ch{\mu}{c}$; on the other hand the maps $c\mapsto
\ch{\mu}{c}$ for $\mu\in\R$ define characters on $\RB$, which
continuously extend the formula \eqref{eq:char_form} from $\R$ to $\RB$.

Both $\RB$ and $\RBd$ carry Haar-measures
$\text{d}c$, $\text{d}\mu$: $\text{d}\mu$ is just
the counting measure on $\R$,
\begin{equation*}
\int_{\RBd} \ft{f}{\mu}\, \text{d}\mu =
\sum_{\mu\in\R} \ft{f}{\mu}
\end{equation*}
and $\text{d}c$ is characterized by
\begin{equation*}
\int_{\RB} \ch{\mu}{c}\, \text{d}c=
\delta_{\mu,0}.
\end{equation*}
Fourier transformation
\begin{equation*}
\label{eq_gft}
\ft{f}{\mu'}\doteq \int f(c) \ch{-\mu'}{c}
\,\text{d}c
\end{equation*}
is an isomorphism $L^2(\RB,\text{d}c)\rightarrow
L^2(\RBd,\text{d}\mu)$. The characters
$\ch{\mu}{\cdot}$ form an uncountable orthonormal basis in $\mathcal{H} =
L^2(\RB,\text{d}c)$, which is therefore inseparable. We will also use
the Hilbert spaces $L^2(\RB\times\RBd,\text{d}c\,\text{d}\mu)$ and
$L^2(\RBd\times\RBd,\text{d}\mu\,\text{d}\mu)$ which are isomorphic
under the partial Fourier transform in the first variable
\begin{equation*}
\widehat{F}(\nu,\mu) = \int_{\RB} F(c,\mu)\ch{-\nu}{c}\,\text{d}c.
\end{equation*}

Let us define some further function spaces (which roughly correspond to
the Schwartz spaces occuring in the theory on $\R$).
\begin{defi}
Denote by
\begin{align*}
\cyl &:\quad \text{the \emph{finite} span of
characters on $\RB$},\\
\cyld &:\quad\text{the image of $\cyl$
under Fourier transform},\\
\cylcyld &:\quad\text{the algebraic tensor product $\cyl\otimes\cyld$.}
\end{align*}
\end{defi}

Some remarks about this definition: first, any element $\Psi\in\cyl$ may be
written as a finite sum
\begin{equation*}
\Psi = \sum_{\mu\in\R} \ft{\Psi}{\mu} h_\mu
\end{equation*}
where the Fourier coefficients
\begin{equation*}
\ft{\Psi}{\mu} = \int \Psi(c)\overline{h_\mu(c)}\,dc
\end{equation*}
of $\Psi$ vanish for all but finitely many $\mu\in\R$. Accordingly,
$\cyld$ consists of all complex-valued functions on $\R$ which are
nonzero only at finitely many points. Second, $\cylcyld$,
the finite span of functions $f_1\otimes f_2$ on
$\RB\times\RBd$ with $f_1\in\cyl$, $f_2\in\cyld$, can be described
equivalently as
the set of functions $f(c,\lambda)$ on $\RB\times\RBd$ that
are in $\cyl$ for fixed $\lambda$ and in $\cyld$
for fixed $c$. $\cylcyld$ is a $*$-algebra under pointwise linear
combination, products and complex conjugation; its elements are
absolutely integrable with respect to the product measure $dc\,d\mu$ on
$\RB\times\RBd$ and hence can be integrated as nested integrals in
either order.

We shall also make use of the algebraic duals of these spaces, which can
be easily characterised: the dual $\cyldd$ of $\cyld$ consists of all
functionals $f:\cyld\to\C$ of the form
\begin{equation*}
f(\Xi)\doteq \sum_\mu f_\mu\Xi_\mu, \qquad(\Xi\in\cyld)
\end{equation*}
where $\mu\mapsto f_\mu$ is any complex-valued function on $\R$; while
$\cyl^*$ is the image of $\cyldd$ under the dual of the Fourier
transform. Thus $\cyl^*$ consists of all functionals $\Gamma:\cyl\to\C$ of
the form
\begin{equation*}
\Gamma(\Phi) \doteq \sum_\mu \widehat{\Gamma}_\mu \widehat{\Phi}_{-\mu}
\qquad(\Phi\in\cyl)
\end{equation*}
where the Fourier coefficients $\widehat{\Gamma}_\mu=\Gamma(h_{-\mu})$ form
an arbitrary complex-valued function $\mu\mapsto \widehat{\psi}_\mu$ on
$\R$. The Fourier transform then extends to a map from $\cyl^*$ to
$\cyldd$ so that
\begin{equation*}
\widehat{\Gamma}(\Xi) = \sum_\mu \widehat{\Gamma}_\mu \Xi_\mu \qquad(\Xi\in\cyld),
\end{equation*}
whereupon the Parseval identity holds in the form
\begin{equation*}
\Gamma(\overline{\Phi}) = \widehat{\Gamma}(\overline{\widehat{\Phi}})\qquad(\Gamma\in\cyl^*,\Phi\in\cyl).
\end{equation*}

There is a particular class of distributions over
$\cyl$ whose action can be expressed in terms of the restriction of
cylindrical functions to the real line.

\begin{lemm}
\label{le_reduction}
Suppose $\Gamma\in\cyl^*$ has the property that
$\widehat{\Gamma}:\mu\mapsto\widehat{\Gamma}_\mu$ is the Fourier
transform of a finite complex measure $\rho$ on $\R$,
\begin{equation*}
\widehat{\Gamma}_\mu = \int_\R e^{-i\mu x} \text{d}\rho(x)
\end{equation*}
(in particular, this holds if $\mu\mapsto\widehat{\Gamma}_\mu$ is a Schwartz
function on $\R$).
Then the action of $\Gamma$ on any $\Psi\in\cyl$ is
\begin{equation}\label{eq_reduction_identity}
\Gamma(\Psi) = \int_\R \Psi|_\R(x) \, \text{d}\rho(x)
\end{equation}
which may be written
\begin{equation*}
\Gamma(\Psi) = \int_\R (\mathcal{F}^{-1}\widehat{\Gamma})(x)\Psi|_\R(x) \, \text{d}x
\end{equation*}
if $\mu\mapsto\widehat{\Gamma}_\mu$ is of Schwartz class,
where $\mathcal{F}$ is the usual Fourier transform on $\R$, $(\mathcal{F}\psi)(p)=\ftn{\psi}(p)$, and the
measure $\text{d}x$ is the usual Lebesgue measure on $\R$.
\end{lemm}
\begin{proof} By linearity it is enough to consider the case
$\Psi(c)=\ch{-\mu}{c}$, for which $\Psi|_\R(x)=e^{-i\mu x}$. By
definition, the left-hand side of \eqref{eq_reduction_identity} is
$\widehat{\Gamma}_\mu$, and the result follows.
\end{proof}
In other words, this type of distribution over $\cyl$ acts on a
cylindrical function as integration of the cylindrical function
restricted to $\R$ with respect to a measure on $\R$.

For later use, we note a simple application of Bochner's theorem.
\begin{lemm}
\label{le_Bochner}
Under the hypotheses of Lemma~\ref{le_reduction}, if $\Gamma$ is
positive (i.e., $\Gamma(\Psi)\ge 0$ for all pointwise nonnegative
$\Psi\in\cyl$) then $\rho$ is a finite positive measure.
\end{lemm}
\begin{proof} Considering any cylindrical function of the form $\Psi(c)=\left|\sum_{i=1}^N
\xi_i h_{\lambda_i}(c)\right|^2$, and applying Lemma~\ref{le_reduction}, we see that
\begin{equation*}
0\le \Gamma(\Psi) = \sum_{i,j=1}^N \overline{\xi}_i\xi_j \int_\R e^{i(\lambda_i-\lambda_j)x}\text{d}\rho(x)
\end{equation*}
Thus the Fourier transform of $\rho$ is a function of positive type, and
hence $\rho$ is a positive measure by Bochner's theorem (Theorem IX.9 in
\cite{ReedSimonII}).
\end{proof}

Finally we define the operators $\op{h}_\mu$ ($\mu\in\R$) and $\op{p}$
by their actions
\begin{equation*}
(\op{h}_\mu \Psi)(c)=\ch{\mu}{c} \Psi(c),\qquad
(\op{p}\Psi)(c) =\sum_\mu \mu \ft{\Psi}{\mu} h_{\mu}(c)
\end{equation*}
on wave functions $\Psi\in\cyl$; each $\op{h}_\mu$ extends to a bounded
operator on $L^2(\RB)$.
We note that $\op{h}_\mu$ and $\op{p}$ have the same
commutation relations as $\exp[i\mu\op{X}]$ and
$\op{P}$, $\op{X}$ and $\op{P}$ being the operators of the
Schr{\"o}dinger representation from
\eqref{eq_schroedinger}. In the latter case, however, the generator
$\op{X}$ may be recovered by differentiation
\begin{equation*}
\op{X}\Psi = - i \left.\frac{d}{d\mu}
\exp[i\mu\op{X}]\Psi\right|_{\mu=0};
\end{equation*}
here, however, $\mu\mapsto \op{h}_\mu\Psi$ is not differentiable on a
dense domain of $\Psi$ in $L^2(\RB)$. In other words, there is no operator on $L^2(\RB)$ corresponding
to the position operator $\op{X}$.

\subsection{The Wigner transform on $\RB$}
We now come to the definition of the Wigner
transform in this setting. At first glance, it is not clear how
to generalize the standard definition \eqref{eq_wig1}
on $\R$, because it is not clear how to divide an element
of $\RB$ by two. The equivalent expressions
\eqref{eq_wig3} and \eqref{eq_wig4} do not suffer from this problem. They contain
Fourier transforms, but those can be replaced by
the Fourier transform \eqref{eq_gft} on $\RB$.
\begin{defi}
For states $\Psi,\Psi'\in \cyl$ the Wigner transform is defined as a
complex-valued function on $\RB\times\RBd$ by
\begin{equation}
\label{eq_def} W(\Psi,\Psi')(c,\mu)\doteq
\int_{\RBd}
\ftcc{\Psi}{\mu'}\ft{\Psi'}{2\mu-\mu'}
\ch{2(\mu-\mu')}{c}\,\text{d}\mu'
\end{equation}
and may be written equivalently as
\begin{equation*}
W(\Psi,\Psi')(c,\mu)=
\int_{\RBd}
\ftcc{\Psi}{\mu-\nu/2}\ft{\Psi'}{\mu+\nu/2}
\ch{\nu}{c}\,\text{d}\nu.
\end{equation*}
\end{defi}
We remark that the second of these expressions is a direct analogue of
\eqref{eq_wig4} (apart from the normalising $2\pi$ factor), while the first is analogous to \eqref{eq_wig3}; the
factor of $2$ appearing in the latter expression arises from a Jacobian
determinant that is not needed in the present setting.
While this definition seems reasonable, its merits should
ultimately be found in its properties. So let us look at some of those, next.
\begin{prop}
The Wigner transform is a sesquilinear map
\begin{equation}
W:\cyl \times \cyl \longrightarrow \cylcyld \label{eq_map1}
\end{equation}
and extends to maps between the following spaces:
\begin{equation*}
W:\begin{array}{ccl} L^2(\RB) \times L^2 (\RB) &\longrightarrow & L^2 (\RB \times
\RBd)\\
\cyl^*\times\cyl^* &\longrightarrow &\cylcyld^*.
\end{array}
\end{equation*}
On $L^2(\RB)\times L^2(\RB)$ the Wigner transform has the overlap
property
\begin{equation}
\label{eq_norm} \iint \overline{W(\Psi_1,\Psi_2)}
W(\Phi_1,\Phi_2)\, \text{d}c \text{d}\mu
=\overline{\scpr{\Psi_1}{\Phi_1}}\scpr{\Psi_2}{\Phi_2}.
\end{equation}
and is hermitean:
\begin{equation}
\label{eq_herm}
\overline{W(\Psi,\Psi')}=W(\Psi',\Psi).
\end{equation}
Furthermore, for $\Psi\in\cyl$,
\begin{equation}
\label{eq_marginal2} \int W(\Psi,\Psi)(c,\mu)\,
\text{d}c =\betr{\ftn{\Psi}_\mu}^2,\qquad \int
W(\Psi,\Psi)(c,\mu)\, \text{d}\mu
=\betr{\Psi}^2(c).
\end{equation}
\end{prop}
\begin{proof}
First look at the assertion \eqref{eq_map1}. For $\Psi,\Psi'\in\cyl$ and any given $\mu$,
$\ftcc{\Psi}{\mu'}\ft{\Psi'}{2\mu-\mu'}$ is
nonzero for only finitely many $\mu'$, hence the
integral in \eqref{eq_def} amounts to a finite
sum, and thus for fixed $\mu$,
$W(\Psi,\Psi')(c,\mu)$ is in $\cyl$. On the other
hand $\ftcc{\Psi}{\mu'}\ft{\Psi'}{2\mu-\mu'}$ is
nonzero only for finitely many $\mu$, with $\mu'$
held fixed. So the integral amounts to a finite
sum of terms with finite support in $\mu$.  This proves \eqref{eq_map1}.

The extension to a map from $\cyl^*\times\cyl^*$ into $\cylcyld^*$ is
obtained as follows: for $\Psi,\Psi'\in\cyl$, and $F\in \cylcyld$ it is
easy to calculate
\begin{equation*}
\int_{\RB\times\RBd} W(\Psi,\Psi')(c,\mu) F(c,\mu) \text{d}c \text{d}\mu
=\int \ftcc{\Psi}{\mu+\nu/2}\ft{\Psi'}{\mu-\nu/2} \widehat{F}(\nu,\mu)
\text{d}\nu \text{d}\mu
\end{equation*}
where $\widehat{F}\in\cyld\otimes\cyld$ is the partial Fourier transform
\begin{equation*}
\widehat{F}(\nu,\mu)=\int_{\RB} F(c,\mu) \ch{-\nu}{c} \text{d}c
\end{equation*}
Noting that above expression converges even if $\Psi$ and $\Psi'$ are
replaced by elements of $\cyl^*$, we then define $W(\Gamma,\Gamma')$ for
$\Gamma,\Gamma'\in\cyl^*$ as the element of $\cylcyld^*$ with action
\begin{equation}\label{eq_dist_Wig_action}
W(\Gamma,\Gamma')[F] = \int \ftcc{\Gamma}{\mu+\nu/2}\ft{\Gamma'}{\mu-\nu/2} \widehat{F}(\nu,\mu)
\text{d}\nu \text{d}\mu \qquad (F\in\cylcyld).
\end{equation}
Short calculations shows that we have the properties
\begin{equation}\label{eq_Wig_ip_dist}
W(\Gamma,\Gamma')[\overline{W(\Psi,\Psi')}] =
\overline{\Gamma}(\Psi)\Gamma'(\overline{\Psi'})
\end{equation}
and
\begin{equation}\label{eq_wignerft}
\widehat{W}(\Gamma,\Gamma')(\nu,\mu)\doteq
W(\Gamma,\Gamma')[h_{-\nu}\otimes\delta_\mu]=
\ftcc{\Gamma}{\mu-\nu/2}\ft{\Gamma}{\mu+\nu/2}
\end{equation}
for $\Gamma,\Gamma'\in\cyl^*$, $\Psi,\Psi'\in\cyl$.

Restricting $\Gamma,\Gamma'$ to $L^2(\RB)$ (regarded as a subspace of
$\cyl^*$) it is easy to see that $\widehat{W}(\Gamma,\Gamma')(\nu,\mu)$
is square summable, so $\widehat{W}(\Gamma,\Gamma')\in
L^2(\RBd\times\RBd)$ is the partial Fourier transform of an element
$W(\Gamma,\Gamma')$ of $L^2(\RB\times\RBd)$. Thus the Wigner transform
maps $L^2(\RB)\times L^2(\RB)$ to $L^2(\RB\times\RBd)$, and \eqref{eq_norm} holds.

Properties \eqref{eq_herm} and \eqref{eq_marginal2} are
confirmed by a short calculation. The latter also holds modulo technical
refinements in the case $\Psi\in L^2(\RB)$, but we will not pursue this here.
\end{proof}
\subsection{Positivity properties}

In the standard setting, Hudson's theorem shows that the Wigner function
is not a probability distribution except for Gaussian states. Our
purpose in this subsection is to investigate this issue for the Wigner
functions of elements of $\cyl$ and $\cyl^*$. Again, Gaussians will play
an important role; however, these must now be treated as distributions
because they are not elements of $L^2(\RB)$. We define them as follows:
\begin{defi}
$\Gamma\in\cyl^*$ is called Gaussian if its
Fourier transform is of the form
\begin{equation}
\label{eq_gauss} \ft{\Gamma}{\mu}=\exp
[-a\mu^2+b\mu+c]
\end{equation}
where $a,b,c \in \C$ and Re$(a)>0$.
\end{defi}
This notion is justified because, according to
Lemma \ref{le_reduction},
\begin{equation}
\label{eq_g} \Gamma(\Psi)= \frac{1}{\sqrt{4 \pi
a}}\int_\R \exp[-(x-ib)^2/(4a) +c]\,
\Psi|_{\R}(x)\,\text{d}x
\end{equation}
for any
cylindrical function $\Psi$.

Just as in the standard theory, Gaussians have nice positivity properties:
\begin{prop}
\label{pr_positive}
For $\Gamma$ Gaussian, $W(\Gamma,\Gamma)$ is positive
in the sense that $W(\Gamma,\Gamma)[F]\ge 0$ for any pointwise nonnegative
$F\in\cylcyld$. Furthermore, equality holds if and only if $F=0$.
\end{prop}
\begin{proof}
By linearity it is enough to prove this for $F$ of the form
$F=\Phi\otimes\delta_{\mu_0}$ for pointwise nonnegative $\Phi\in\cyl$.
Let $\Gamma$ be Gaussian as in \eqref{eq_gauss}. Then
\begin{align*}
\ftcc{\Gamma}{\mu+\nu/2}\ft{\Gamma}{\mu-\nu/2} &=
\exp\left[-a(\mu-\frac{\nu}{2})^2 -\overline{a}(\mu+\frac{\nu}{2})^2\right.\\
&\qquad\qquad\left. +b(\mu-\frac{\nu}{2})+\overline{b}(\mu+\frac{\nu}{2})
+2\re c\right]\\
&= f(\mu,a,b,c) \exp\left[-\frac{\re a}{2}\nu^2+i(2\im(a)\mu- \im (b))\nu\right]
\end{align*}
where
\begin{equation*}
f(\mu,a,b,c)= \exp[-2\re(a)\mu^2+2\re(b)\mu +2\re(c)]
\end{equation*}
is a \textit{positive} expression that does not depend on $\nu$, while the remaining factor is
in the Schwartz class. Using \eqref{eq_dist_Wig_action} and
$\widehat{F}(\nu,\mu)=\ft{\Phi}{\nu}\delta_{\mu\mu_0}$, we therefore have
\begin{equation}
W(\Gamma,\Gamma)[F] = f(\mu_0,a,b,c)\sum_{\nu\in\R} \exp\left[-\frac{a}{2}\nu^2+i (2\im(a)\mu_0-\im (b))
\nu\right]\widehat{\Phi}_\nu
\end{equation}
and applying Lemma \ref{le_reduction}, we have:
\begin{equation*}
W(\Gamma,\Gamma)[F]=\frac{f(\mu_0,a,b,c)}{\sqrt{2\pi\re(a)}}
\int_{\R}\exp\left[-\frac{(x+2\im(a)\mu_0-\im(b))^2}{2\re(a)}\right]\, \Phi|_{\R}(x)\,dx
\end{equation*}
which is manifestly positive as $\Phi\ge 0$, and vanishes if and only if
$\Phi=0$. Every pointwise positive element of $\cylcyld$ is a convex
combination of functions of the above form, which completes the proof.
\end{proof}

We now present a converse to this result, which is analogous to Hudson's
theorem (Prop.~\ref{pr_gauss}) and indeed makes use of the classical
result. The hypotheses can be weakened further, but we do not pursue
this for simplicity.

\begin{prop}\label{pr_Hudson}
Suppose $\gamma\in L^2(\R)\cap L^1(\R)\backslash\{0\}$ and define $\Gamma\in\cyl^*$
by $\widehat{\Gamma}_\mu=\ftn{\gamma}(\mu)$.
If $W(\Gamma,\Gamma)$ is positive then $\Gamma$ is a Gaussian.
\end{prop}
\begin{proof}
We begin by noting that the usual Wigner function of $\gamma$, $W(\gamma,\gamma)(x,p)$,
has the property that
$W(\gamma,\gamma)(\cdot,\mu) \in L^1(\R)$ for each $\mu$ because $\gamma\in
L^1(\R)$, and is also continuous because $\gamma\in L^2(\R)$. Next,
observe that
\begin{align*}
\ftcc{\Gamma}{\mu+\nu/2}\ft{\Gamma}{\mu-\nu/2} &=
\int_\R e^{i(x+y)\nu/2} e^{i\mu (x-y)} \overline{\gamma}(x)\gamma(y)
\,\text{d}x\,\text{d}y \\
&=\int_\R e^{-ix\nu}W(\gamma,\gamma)(-x,\mu)\text{d}x
\end{align*}
which for each fixed $\mu\in\R$ exhibits the left-hand side as the Fourier transform of
a measure obtained from $W(\gamma,\gamma)$; this measure is finite
by the $L^1$ property mentioned above.
Accordingly, by Lemma~\ref{le_reduction},
\begin{equation}
W(\Gamma,\Gamma)[\Phi\otimes\delta_\mu] = \int_\R
\Phi|_\R(x) W(\gamma,\gamma)(-x,\mu)\text{d}x
\end{equation}
for any cylindrical function $\Phi$. As the left-hand side is
nonnegative for every $\mu\in\R$ and each positive $\Phi\in\cyl$,
it follows by Lemma~\ref{le_Bochner} and continuity that $W(\gamma,\gamma)$ is pointwise nonnegative.
As $\gamma\in L^2(\R)\backslash\{0\}$, Hudson's theorem entails that $\gamma$ and
hence $\ftn{\gamma}$ are Gaussian. Thus $\Gamma$ is a Gaussian element of
$\cyl^*$.
\end{proof}

We can also examine the positivity properties of Wigner functions of
cylindrical functions. Although $\cyl$ may be embedded in $\cyl^*$,
the resulting distributions do not satisfy the hypotheses of
Prop.~\ref{pr_Hudson}. Indeed, we may immediately observe that the
statement of Prop.~\ref{pr_Hudson} cannot hold for all elements of
$\cyl$:
\begin{lemm}
\label{le_weird}
The Wigner function for a pure
character $\Phi(c)\doteq a \ch{\mu}{c}$ is
positive,
\begin{equation*}
W(\Phi,\Phi) =|a|^2\delta_{\mu_0}(\mu).
\end{equation*}
\end{lemm}
\begin{proof}
Let $\Phi=h_{\mu_0}$. Then
\begin{equation}
W(\Phi,\Phi) = \int_{\RBd}
\ftcc{\Phi}{\mu'}\ft{\Phi}{2\mu-\mu'}
\ch{2(\mu-\mu')}{c}\,\text{d}\mu'.
\end{equation}
The first factor in the integrand vanishes unless $\mu'=\mu_0$, so
\begin{equation}
W(\Phi,\Phi) = \ftcc{\Phi}{\mu_0}\ft{\Phi}{2\mu-\mu_0}
\ch{2(\mu-\mu_0)}{c}
\end{equation}
and we see that the second factor in this expression vanishes unless
$2\mu-\mu_0=\mu_0$, i.e., $\mu=\mu_0$. Thus
\begin{equation}
W(\Phi,\Phi)(c,\mu) = |\widehat{\Phi}_{\mu_0}|^2\ch{0}{c}\delta_{\mu_0}(\mu)
=|a|^2\delta_{\mu_0}(\mu)\ge 0,
\end{equation}
so $W(\Phi,\Phi)$ is a positive element of $\cylcyld$.
\end{proof}
This is at first very surprising. After all, the
characters are rather quantum mechanical states
with `infinite uncertainty' for multiplication
operators on $\RB$. On the other hand, they are
eigenstates for the operator $\op{p}$, and in fact one
should interpret them as `degenerate Gaussian
states' as follows: In the usual setting of
analysis on $\R$, consider a Gaussian on Fourier
space and its Fourier transform on position
space. We are interested in the limit of bringing
its width in Fourier space to zero, while keeping
its integral fixed. The limit is not a square
integrable function, neither in Fourier- nor in
position space: In Fourier space, it is the delta
distribution, in position space it is a function
of constant modulus. For analysis on $\RB$, the
situation is however drastically different. The
limit is well defined, giving a Kronecker delta
on $\RBd$ and a character on $\RB$. Given this,
the statement of Lemma \ref{le_weird} is perhaps
less surprising.

Next we present a converse
to Prop. \ref{pr_positive}:
\begin{prop} Suppose that $\Phi\in\cyl$. The following are
equivalent:
\begin{enumerate}
\item $W(\Gamma,\Gamma)(\overline{W(\Phi,\Phi)})>0$ for all Gaussians
$\Gamma\in\cyl^*$;
\item $\Phi$ is of the form $\Phi=ah_{\mu}$ for some
$\mu$, i.e., a scalar multiple of a character.
\end{enumerate}
In particular, $W(\Phi,\Phi)$ is pointwise nonnegative for $\Phi\in\cyl$ if and only if
$\Phi$ is a scalar multiple of a character.
\end{prop}
\begin{proof} We adapt the standard proof of Hudson's
theorem \cite{hudson}. First observe that with
the definitions of Wigner function for
$\Gamma\in\cyl^*$ and Fourier transform we have
\begin{equation}
W(\Gamma,\Gamma)(\overline{W(\Phi,\Phi)}) = |\Gamma(\overline{\Phi})|^2
\end{equation}
by \eqref{eq_Wig_ip_dist}. So property 1 implies that
$\Gamma(\overline{\Phi})$ is nonvanishing for any Gaussian $\Gamma$.

Now consider the family of Gaussians $\Gamma_z$ with Fourier transform
$(\widehat{\Gamma_z})_\mu = \exp(-\mu^2 -i\mu z)$. We have
\begin{equation}
G(z)\doteq \Gamma_z(\overline{\Phi}) = \sum_\mu \overline{\widehat{\Phi}}_\mu \exp(-\mu^2 -i\mu z)
\end{equation}
which is clearly an entire function of exponential type
(i.e., $|G(z)|\le Ae^{B|\textrm{Im}\, z|}$ for constants $A$, $B$).
As it is also nonvanishing by
the previous observation we may conclude by a result of
Hadamard (Theorem VIII.10 in \cite{Saks_Zygmund}) that $G(z)$ is the exponential
of a polynomial of at most first degree, i.e.,
\begin{equation}
G(z) = a e^{ibz}
\end{equation}
for complex constants $a\not=0$ and $b$.
Now the restriction of $G(z)$ to the real line is
bounded, so we may conclude that $b$ is real.

We therefore have
\begin{equation}
\sum_\mu \overline{\widehat{\Phi}}_\mu \exp(-\mu^2 -i\mu t) = a e^{ibt}
\end{equation}
for all real $t$, from which it follows that $\widehat{\Phi}_\mu$ is
nonzero only for $\mu=b$. Accordingly $\Phi$ is a scalar multiple of the
character $h_{b}$. Thus $1\implies 2$.

In the converse direction we set $\Phi=a
h_{\mu_0}$ and refer to Lemma \ref{le_weird},
which tells us that $W(\Phi,\Phi)$ is a positive
element of $\cylcyld$. Using Prop.
\ref{pr_positive} we therefore have
$W(\Gamma,\Gamma)(\overline{W(\Phi,\Phi)})>0$ for all
Gaussians $\Gamma$, and the equivalence of 1 and 2 is
established. The proof is concluded by remarking that if $W(\Phi,\Phi)$
is pointwise positive and not identically zero then $W(\Gamma,\Gamma)[W(\Phi,\Phi)]>0$
for all Gaussians $\Gamma$ by Prop.~\ref{pr_positive}, and hence $\Phi$
is a scalar multiple of a character; the converse is given by
Lemma~\ref{le_weird}. It is also trivial that $W(\Phi,\Phi)$ is
identically zero if and only if $\Phi=0$.
\end{proof}

It is worth remarking that the essential difference between this result
and the standard line of argument is that the support of
$\widehat{\Phi}$ is supposed to be bounded, which permits us to obtain
an exponential bound of first order. In the usual proof of Hudson's
theorem one does not have this luxury and the quadratic bound arises by
completing a square to bound the Gaussian term.

\subsection{Quantization}
Now we discuss the Wigner transform and
quantization. Consider a distribution
$\sigma\in\cylcyld^*$. It defines a
sesquilinear form $B_\sigma$ on $\cyl$ via
\begin{equation*}
B_\sigma(\Psi_1,\Psi_2) =\sigma(W(\Psi_1,\Psi_2)).
\end{equation*}
Motivated by Prop.\ \ref{pr_calculus} we ask the following question:
When is the form $B_\sigma$ induced by an operator $\op{\sigma}$
on $\cyl$? A partial answer to this question can be given as follows.

Let $\sigma$ be in $\cylcyld$.  Then we can
\textit{define} an operator $\op{\sigma}$ on $\cyl$ by
\begin{equation}
\label{eq_opdef} (\op{\sigma}\Psi)(c) \doteq
\iint_{\RBd\times\RBd}
\ft{\Psi}{\nu}\widehat{\sigma}(\mu-\nu,(\nu+\mu)/2)\ch{\mu}{c}
\,\text{d}\mu\text{d}\nu.
\end{equation}
This definition is justified by the following fact.
\begin{lemm}
The matrix elements of the operator $\op{\sigma}$ as defined by \eqref{eq_opdef}
are given by
\begin{equation*}
\scpr{h_\mu}{\op{\sigma}h_\nu}=B_\sigma(h_\mu,h_\nu)\qquad ( \equiv \sigma(W(h_\mu,h_\nu))).
\end{equation*}
In particular, $\op{\sigma}$ is
symmetric if $\sigma$ is real.
\end{lemm}
The proof is straightforward, noting that
\begin{equation}
\label{eq_matrix}
B_\sigma(h_\mu,h_\nu)=\widehat{\sigma}(\mu-\nu,(\nu+\mu)/2).
\end{equation}
This answers the question, albeit only for symbols $\sigma$ in $\cylcyld$.

What about more general symbols? If $\sigma$ is only in
$\sigma\in\cylcyld^*$ it is not clear
whether the integration in \eqref{eq_opdef}
converges or not, so this seems to be too general. On the other hand there are
functions $\sigma$ that are not even in
$L^2(\RB\times\RBd)$, for which \eqref{eq_opdef} \emph{does}
make sense. For our purposes, it will be sufficient to restrict to
the symbol classes defined in the following result:
\begin{prop}
(a) Let $\symb \subset \cylcyld^*$ be the set of distributions $\sigma$
such that
\begin{equation}
\label{eq_crit}
\text{for any
fixed } \beta\in \RBd: \quad
M_{\alpha\beta}\doteq\widehat{\sigma}(\alpha-\beta,(\alpha+\beta)/2) \text{ is in  L}^1(\RBd).
\end{equation}
Then for any $\sigma\in\symb$, \eqref{eq_opdef} defines a (possibly
unbounded) operator $\op{\sigma}$ with
domain $\cyl$.\\
(b) Let $\symb_\infty\subset\symb$ be the set of distributions $\sigma$
for which there exist constants $A,B\ge 0$ such that
\begin{equation}
\sum_\beta |M_{\alpha\beta}|\le A \quad \text{for all $\alpha\in\RBd$}
\end{equation}
and
\begin{equation}
\sum_\alpha |M_{\alpha\beta}|\le B \quad \text{for all $\alpha\in\RBd$}.
\end{equation}
For each such $\sigma$, the operator $\op{\sigma}$ extends to a bounded operator
(also denoted $\op{\sigma}$) on $L^2(\RB)$ with
\begin{equation*}
\|\op{\sigma}\|\le \sqrt{AB}.
\end{equation*}
Moreover, the adjoint $\op{\sigma}^*$ is the quantization of
$\overline{\sigma}$.
\end{prop}
\begin{proof} (a) is trivial; while (b) is immediate from the Schur
test (Theorem 5.2 in~\cite{HalmosSunder}). The statement about the
adjoint follows on noting that
\begin{equation*}
B_{\overline{\sigma}}(h_\mu,h_\nu)=
\widehat{\overline{\sigma}}(\mu-\nu,(\mu+\nu)/2)
= \overline{\widehat{\sigma}(\nu-\mu,(\mu+\nu)/2)}
=\overline{B_{\sigma}(h_\nu,h_\mu)}.
\end{equation*}
\end{proof}
We remark that for example any function
$\sigma(c,\lambda)$ on $\RB\times\RBd$ which is
in $\cyl$ for fixed $\lambda$ belongs to $\symb$.

What does the quantization \eqref{eq_opdef} give?
Straightforward calculations show the following
\begin{lemm}
For $\sigma_1(c,\lambda)=\ch{\mu}{c}$ and
$\sigma_2(c,\lambda)=\lambda$
\begin{equation*}
\op{\sigma}_1 =\op{h}_\mu, \qquad
\op{\sigma}_2 = \op{p}.
\end{equation*}
For $\sigma_3(c,\lambda)=\lambda \ch{\mu}{c}$
\begin{equation*}
\op{\sigma_3}  = \frac{1}{2}(\op{h}_\mu \op{p}+ \op{p} \op{h}_\mu).
\end{equation*}
\end{lemm}
We note that the latter is the totally symmetric
ordering of $p\exp(i\mu x)$, thus one should
think of the quantization given by
\eqref{eq_opdef} as Weyl quantization. Weyl
ordering is interesting because of its properties
(see \cite{folland} for a discussion), and it has
become important in loop quantum cosmology
because it was used in a novel method for
obtaining effective equations of motion from the
quantum theory
\cite{Bojowald:2005cw,Bojowald:2006ww}.
\section{Application to the quantization of the Hamiltonian constraint}
\label{se_applications}
As we have pointed out above, the Wigner function
can be used to quantize phase space functions in
a systematic way. In the following we will apply
this technique in a case that is of importance to
LQC. More generally we expect that this formalism
will be useful in the context of obtaining
effective equations of motion from the quantum
theory \cite{Bojowald:2005cw,Bojowald:2006ww},
and for cases in which complicated phase space
functions have to be quantized.

The example we discuss here is important in the
context of describing homogenous and isotropic
cosmology coupled to a scalar field in the
framework of LQC \cite{Ashtekar:2006uz}. It was
found there that the standard way of quantizing
the Hamiltonian constraint
(e.g.~\cite{Ashtekar:2003hd}) led to physically
unacceptable results, and a new quantization was
introduced in \cite{Ashtekar:2006wn}. In very
brief terms, it can be described as follows: As
customary in LQC, before quantization the
curvature is expressed in terms of the connection
along a small edge. Formerly this edge was taken
to have a length proportional to the smallest
quantum of length in the full theory, as measured
in a fiducial background metric. The basic idea
of Ashtekar, Pawlowski and Singh (APS)
\cite{Ashtekar:2006wn} is to determine it in a
similar way, but with respect to the physical
metric which is subject to quantization. On a
technical level this requires quantization of the
phase space function
\begin{equation}
\label{eq_new} e(c,\mu)= \ch{\overline{\mu}(\mu)}{c}
\end{equation}
where $\overline{\mu}$ is a function fulfilling
\begin{equation}
\label{eq_mubar}
\overline{\mu}(\mu)^2=\frac{3\sqrt{3}}{2}\betr{\mu}^{-1}.
\end{equation}
The new symbol $e(c,\mu)$ replaces the function $e_0(c)=\exp(i
\mu_0 c)$ in the old quantization of the
constraint, which is constant in $\mu$.\footnote{Eq.~\eqref{eq_mubar} determines
$\overline{\mu}$ only up to sign. We will choose
$\overline{\mu}$ positive, in agreement with
\cite{Ashtekar:2006wn}.} According to
\cite{Ashtekar:2006wn}, $\mu_0$ was chosen as
$3\sqrt{3}/2$.

Obviously quantization of \eqref{eq_new}
necessitates a choice of ordering. In
\cite{Ashtekar:2006wn} a quantization was arrived
at in the following fashion: Naively $c$ would be
quantized by a derivative in $\mu$. That
derivative, and hence an operator corresponding
to $c$, fail to exist on functions in $\cyld$
however. Nevertheless one can study the action of
the operator
\begin{equation*}
\op{e}_{\text{APS}} \doteq \sum_{n=0}^{\infty} \frac{1}{n!} [
\overline{\mu}(\mu) \frac{d}{d\mu} ]^n
\end{equation*}
on smooth functions on $\R$. This action is given
by pullback with a certain diffeomorphism of
$\R$, and it continues to make sense on functions
in $\cyld$.
The upshot is that APS define an operator $\op{e}_{\text{APS}}$ on
$\cyl$ by
\begin{equation}
\label{eq_unitary} \widehat{(\op{e}_{\text{APS}} \Psi)}_\mu  = \ft{\Psi}{\sign
(\mu')\betr{\mu'}^{\frac{2}{3}}} \quad \text{
with } \quad
\mu'=\sign(\mu)\betr{\mu}^{\frac{3}{2}} +
\frac{1}{K}
\end{equation}
where $K$ is a specific numerical constant.

This quantization is very plausible for a number of reasons. First, $\op{e}_{\text{APS}}$ is a
unitary operator. Second, although \eqref{eq_unitary}
looks very complicated, it can be given a simple
interpretation: $\op{e}_{\text{APS}}$ is a constant shift on
wave functions over the \textit{volume}. Third,
maybe most importantly, it ultimately leads to a
Hamiltonian constraint that is physically viable.

Still, since \eqref{eq_unitary} is at least
partly motivated by reference to a differentiable
structure on $\RBd$, -- something that does not
exist -- it may be interesting to consider
alternatives. Therefore we proceed now to
quantize the same classical function using the
Wigner transform. The symbol $e(c,\mu)$ from
equation \eqref{eq_new} is undefined at
$\mu=0$, and we will remedy this by setting $e(c,0)=0$. Although this
appears {\em ad hoc}, it will be shown below that the same results are
obtained by taking limits of quantized operators formed from regularised
versions of $\overline{\mu}$. With that end in view, let us first
consider general symbols of the form
\begin{equation*}
e_f(c,\mu) = \ch{f(\mu)}{c}
\end{equation*}
where $f:\R\to\R$. Noting that
$\widehat{e}_f(\nu,\mu)=\delta_{\nu,\,f(\mu)}$, we have
\begin{equation*}
\widehat{e}_f(\alpha-\beta,(\alpha+\beta)/2)
= \delta_{\alpha-\beta,\, f((\alpha+\beta)/2)}
\end{equation*}
Accordingly, the properties of the quantisation $\op{e_f}$ are closely
related to the properties of equation
\begin{equation}\label{eq_star}
\alpha-\beta = f\left(\frac{\alpha+\beta}{2}\right).
\end{equation}
In particular, if there are constants $A$ and $B$ such that
\eqref{eq_star} has at most $A$ (resp., $B$) solutions for $\beta$
(resp., $\alpha$) for each fixed $\alpha$ (resp., $\beta$) then $e_f\in\symb_\infty$
and $\|\op{e_f}\|\le \sqrt{AB}$. We will now restrict to functions $f$
for which this condition holds. One may also note that
\begin{equation*}
\op{e_f} h_\beta = \sum_{\alpha\in S_f(\beta)} h_\alpha
\end{equation*}
where $S_f(\beta)$ is the set of $\alpha$ solving \eqref{eq_star} for the
given $\beta$. It follows that $\op{e_f}$ is unitary if and only if
\eqref{eq_star} implicitly defines a bijection $\beta\mapsto\alpha(\beta)$
of $\R$. Indeed, $\op{e}_\text{APS}$ is precisely of this form for a
suitable $f$.

For the particular symbol $e(c,\mu)$ of interest, the above remarks are
valid modulo the special treatment of $\mu=0$; the upshot is that
\begin{equation}
\label{eq_edef}
\op{e} h_\beta = \sum_{\alpha\in
S(\beta)} h_\alpha
\end{equation}
where $S(\beta)$ is the set of solutions
$\alpha\in\R\backslash\{-\beta\}$ to
$\alpha-\beta = \overline{\mu}((\alpha+\beta)/2)$. Equivalently, these
are the solutions to
\begin{equation}\label{eq_outer}
|\alpha+\beta|(\alpha-\beta)^2 =3\sqrt{3}.
\end{equation}
with $\alpha>\beta$; analysis of this equation reveals that there are
one, two, or three solutions for fixed $\beta$ according to whether
$\beta$ is greater than, equal to, or less than $-3^{3/2}2^{5/3}$. (This
is illustrated in the
first diagram in Figure \ref{fi_separate}, in which the
solutions would be the intersections of an
$\beta=$const. line with the graph.) The
same is true for solutions in $\beta$ for fixed $\alpha$; it therefore
follows that $e\in\symb_\infty$ and that we have $\|\op{e}\|=3$. In
contrast to $\op{e}_\text{APS}$, then, $\op{e}$ is not unitary.

The relationship between $\op{e}$ and $\op{e}_\text{APS}$ will be
discussed further below; first, we show how $\op{e}$ may be obtained as
a limit of quantizations based on regularised versions of
$\overline{\mu}$. A function
$f:\R\to\R$ will be called an \emph{$\epsilon$-regularisation} of $\overline{\mu}$
if $f(\mu)=\bar{\mu}(\mu)$ for $|\mu|>\epsilon$ and $f$ is concave
on $|\mu|\le\epsilon$. (Concavity is adopted here for
convenience; much weaker conditions would also suffice.) Thus, for
example, taking $f(\mu)=\overline{\mu}(\epsilon)$ on $|\mu|\le\epsilon$ would give an
$\epsilon$-regularisation, but there are many other possibilities.

\begin{lemm} Let $f$ be any $\epsilon$-regularisation of $\bar{\mu}$.
Then $e_f\in\symb_\infty$ and $\|\op{e_f}\|\le 5$.
\end{lemm}
\begin{proof}
First note that every solution to \eqref{eq_star} with
$|\alpha+\beta|>2\epsilon$ is a solution to \eqref{eq_outer}
with $\alpha>\beta$. We have already seen that there are
at most $3$ solutions to this equation for $\alpha$ (resp.,
$\beta$) at fixed $\beta$ (resp., $\alpha$). It remains to consider
solutions with $|\alpha+\beta|\le 2\epsilon$. Fixing $\beta$, these
solutions correspond to intersections of the graph of $f$ with a straight line, and
there can be at most two of these in this region because $f$ is concave on
$[-\epsilon,\epsilon]$.
The same is true if we fix $\alpha$. Accordingly there are at most $5$
solutions to \eqref{eq_star} on lines of constant $\beta$ or $\alpha$.
The result follows by the foregoing discussion.
\end{proof}

\begin{prop} \label{pr_newholonomy}
Let $f_n$ be any sequence of $\epsilon_n$-regularisations
of $\overline{\mu}$ with $\epsilon_n\to 0^+$. Then the sequence of
operators $\op{e_{f_n}}$ converges strongly to the operator
$\op{e}$ defined above.
\end{prop}
\begin{proof}
As the operators in the sequence are all bounded with norm less than
$5$, it is enough to prove strong convergence on the dense subspace
$\cyl$ of $L^2(\RB)$. In turn, it therefore suffices to establish strong
convergence of the sequence applied to each character $h_\beta$. Fix
$\beta\in\R$ and choose $N$ large enough that $S(\beta)$ has no
intersection with $[-\beta-2\epsilon_n,-\beta+2\epsilon_n]$ for $n>N$.
This is possible because $S(\beta)$ is finite and excludes $-\beta$.
For such $n$, $S_{f_n}(\beta)$ is the union of disjoint sets $S(\beta)$ and
\begin{equation*}
T_n(\beta)\doteq S_{f_n}(\beta)\cap [-\beta-2\epsilon_n,-\beta+2\epsilon_n]
\end{equation*}
so we may write
\begin{equation*}
\op{e_{f_n}} h_\beta =\op{e} h_\beta + \sum_{\alpha\in T_n(\beta)} h_\alpha.
\end{equation*}
We now claim that $T_n(\beta)$ is empty for all sufficiently large $n$,
thus establishing that $\op{e_{f_n}} h_\beta\to \op{e} h_\beta$ and
hence (as $\beta$ is arbitrary) the required result.

The claim is proved as follows: if $\alpha\in T_n(\beta)$ then
$|\alpha+\beta|\le 2\epsilon_n$, and hence
$\alpha-\beta<2(|\beta|+\epsilon_n)$. But by concavity
$f_n((\alpha+\beta)/2)\ge f_n(\epsilon_n)=3^{3/4}(2\epsilon_n)^{-1/2}$ for $|\alpha+\beta|\le
2\epsilon_n$. As $\epsilon_n\to 0$ it is therefore clear that for all
sufficiently large $n$ there are
no solutions to $\alpha-\beta=f_n((\alpha+\beta)/2)$; hence $T_n(\beta)$
is empty as required.
\end{proof}

In view of these results, we are encouraged to regard the choice
$e(c,0)=0$ as well-motivated, and $\op{e}$ as an appropriate Weyl
quantization of the original symbol defined (for $\mu\not=0$) in \eqref{eq_new}.
To visualize the results of the above proposition, and to get a better feeling for the action of
$\op{e}$, it is helpful to plot the space of solutions to
\eqref{eq_star} with $f=\overline{\mu}$. Instead of solving that equation
directly, we will give a parametrization of its
solution space: It is easily checked that
\begin{equation}
\label{eq_solution} \{(\alpha,\beta) |
\alpha-\beta=\overline{\mu}(\frac{\alpha+\beta}{2})\}
=
\{(x+\frac{1}{2}\overline{\mu}(x),x-\frac{1}{2}\overline{\mu}(x)|
x\in\R \}.
\end{equation}
This set is plotted as a graph in the
$\alpha-\beta$-plane in the first diagram of
Figure \ref{fi_separate}.
Many of the properties of $\op{e}$ stated in the
previous lemma can also be obtained from a visual
inspection of what we will call its `matrix
representation'. It will also be useful in the
comparison of $\op{e}$, $\op{e}_{\text{APS}}$, and $\op{h}_{\mu_0}$. Note
that all these operators have only matrix
elements equal to 1 or 0 in the basis $\cyl$. So
we can visualize them by plotting the set of
matrix elements $\scpr{h_\beta}{\,\cdot\,
h_\alpha}$ that are equal to 1 as a set in the
plane.\footnote{In other words, we can plot the
graphs of the operators as the graphs of
functions.}
\begin{figure}
\centerline{\epsfig{file=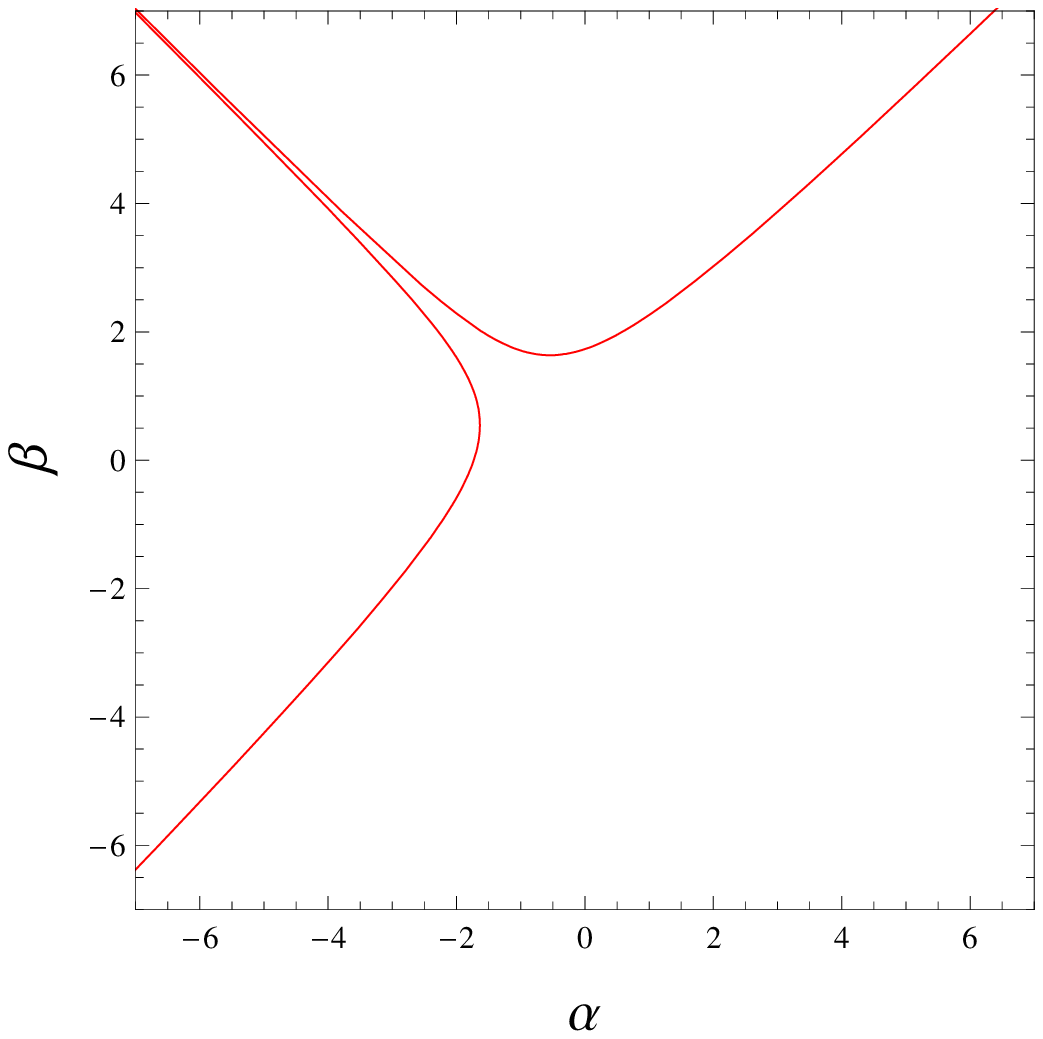,scale=0.5}
\epsfig{file=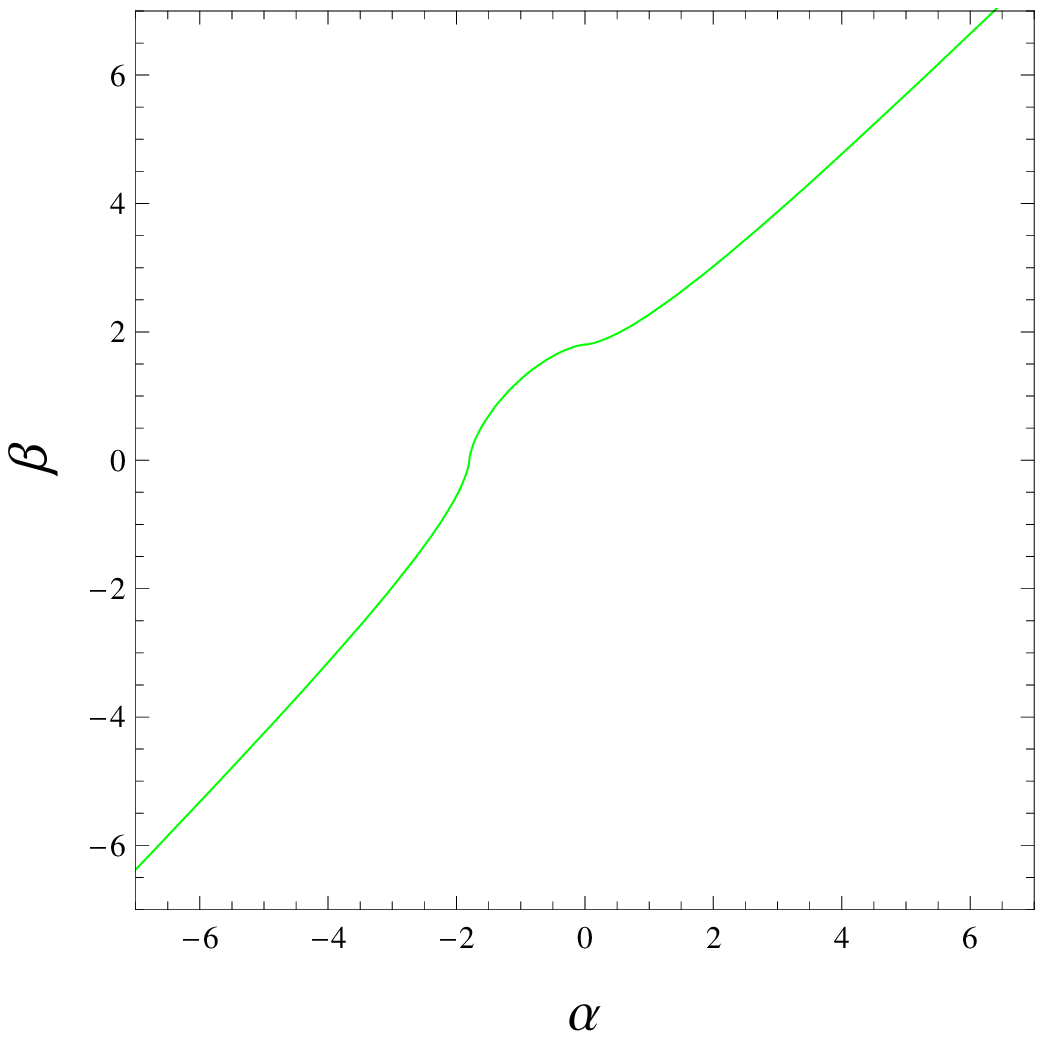,scale=0.5}}
\centerline{\epsfig{file=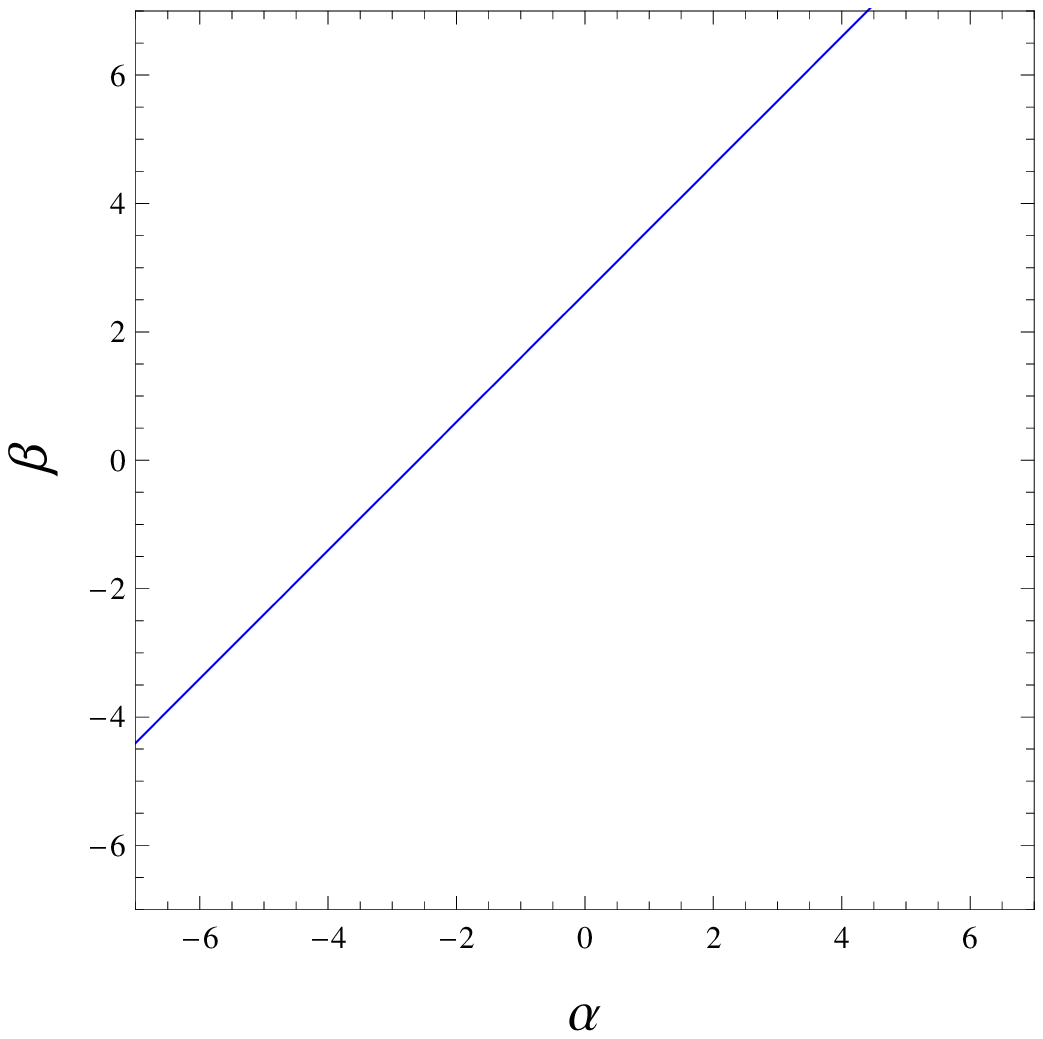,scale=0.5}}
\caption{\label{fi_separate}Matrix
representations of the operators $\op{e}$, $\op{e}_{\text{APS}}$,
$\op{h}_{\mu_0}$}
\end{figure}
\begin{figure}
\centerline{\epsfig{file=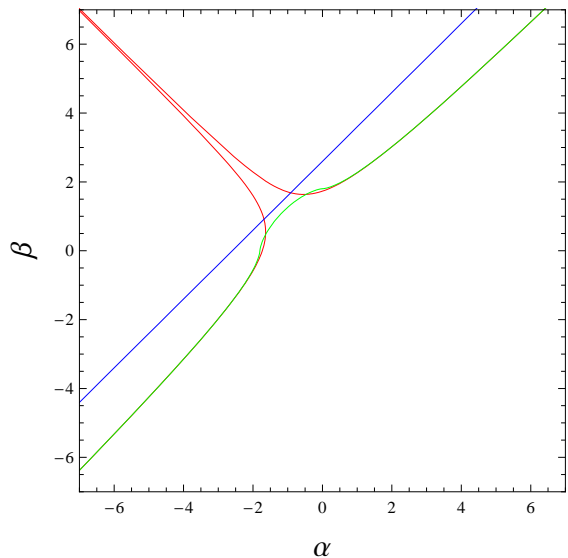,scale=.7}}
\caption{\label{fi_all}Comparison of the
operators $\op{e}$, $\op{e}_{\text{APS}}$, $\op{h}_{\mu_0}$}
\end{figure}
This is done in Figure \ref{fi_separate} for the
operators separately, and in Figure \ref{fi_all},
for easy comparison, into one diagram. One can
see how, for example, $\op{e}$ is not unitary, because
its graph is that of a \textit{multi-valued}
function.

Let us compare the operator $\op{e}$ that we obtained
here with $\op{e}_{\text{APS}}$ of \cite{Ashtekar:2006wn}: The
first difference is that $\op{e}_{\text{APS}}$ unitary in
contrast to $\op{e}$. Another difference is the
``spike'' in the graph for $\op{e}$: Whereas the graph
of $\op{e}$ has a part (solid lines in Fig.
\ref{fi_branches}, that is very similar to $\op{e}_{\text{APS}}$, and can be characterized by
$\alpha\approx\beta$ for large $\alpha$ and
$\beta$, the graph also has a part that is very
different (drawn as a dashed line in Figure
\ref{fi_branches}), that can be characterized by
$\beta \approx -\alpha$, $\alpha<0$. This new
feature may point to difficulties with the
semiclassical limit for this operator because,
loosely speaking, states with large momentum
eigenvalues correspond to a universe with large
spatial extension, and the operator is obviously
changed on those large volume states, as compared
to $\op{e}_{\text{APS}}$ and $\op{h}_{\mu_0}$. On the other hand, in
the present model eigenstates of momentum $\mu$
and $-\mu$ are physically identical\footnote{They
are related by the parity transform $\Pi$ which
is a symmetry of the system, see
\cite{Ashtekar:2006uz} for details.} whence the
two parts of the graph of $\op{e}$ (dashed and solid
respectively, in Fig. \ref{fi_branches}) may act
in a very similar way on the physical level.
\begin{figure}
\centerline{\epsfig{file=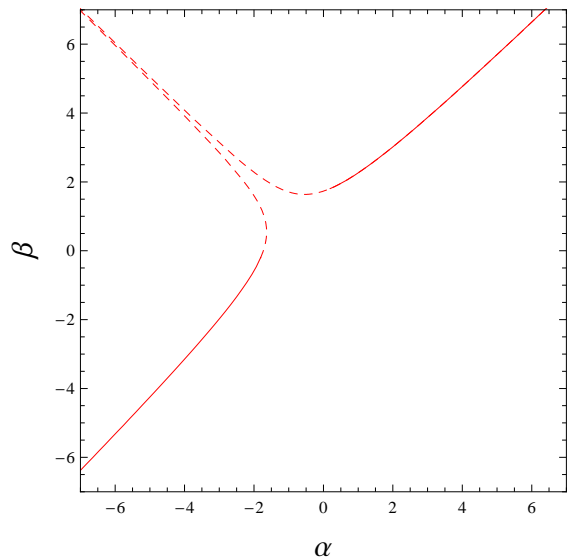,scale=.7}}
\caption{\label{fi_branches} The different parts
of the matrix-representation of $\op{e}$: The part
that differs very little from the operator $\op{e}_{\text{APS}}$
(solid), and the ``spike'' that is rather
different (dashed)}
\end{figure}

Ultimately one will have to construct the full
Hamiltonian constraint using the operator
$\op{e}$ and compare the physical results to
those obtained in \cite{Ashtekar:2006wn}, and we
see no problems of principle for doing this. In
particular we note that using $\op{e}$ will give
a Hamiltonian constraint that commutes with the
action of the parity operator $\Pi$, since one
finds that
\begin{equation*}
\widehat{\sin(\overline{\mu}c)}\,\Pi=-\Pi\,\widehat{\sin(\overline{\mu}c)},\qquad
\widehat{\cos(\overline{\mu}c)}\,\Pi=\Pi\,\widehat{\cos(\overline{\mu}c)}
\end{equation*}
also with this quantization. The analysis will be
substantially more complicated since it does not seem
that there is a basis in which the action of $\op{e}$
drastically simplifies (such as is the case for
$\op{e}_{\text{APS}}$ with respect to the volume eigenvector
basis). In particular it does not seem to be
likely that the kinematical Hilbert space will be
decomposable into different superselection
sectors as was the case in other quantizations.
That said, we will leave a detailed analysis for
the future, and turn now to a discussion of the
results of the present paper.
\section{Closing remarks}
\label{se_discussion}
In the present paper we have given a definition
of the Wigner function for wave functions over
the Bohr compactification $\RB$ of the real line
and shown that our definition possesses many
properties analogous to the Wigner function in
ordinary quantum mechanics.

Since wave functions over the Bohr
compactification figure prominently in loop
quantum cosmology, the Wigner function should be
of use in that context. To demonstrate this, we
used the Wigner function machinery to give an
alternative quantization of an important building
block of the Hamiltonian constraint for
homogenous isotropic cosmology as treated in LQC
\cite{Ashtekar:2006uz,Ashtekar:2006wn}. We should
stress again that the ordering chosen in
\cite{Ashtekar:2006wn} has many desirable
properties and we do not want to claim in any way
that it is wrong or inappropriate. Rather, by
adding Weyl quantization to the toolbox of those
working in LQC, we provide a quantization method
that is applicable to a wide variety of
situations, without the need to make any {\em ad
hoc} choices. Whether the results are physically
viable must still be determined in each instance
separately.

The content of the present paper could be further
developed in several directions: On the
mathematical side, a more detailed investigation
of the properties of the quantization map (e.g.,
with respect to products, or the semiclassical
limit) could be undertaken. On the physical side,
the quantization of the Hamiltonian constraint
using the operator $\op{e}$ from \eqref{eq_edef}
should be completed. Then its physical
implications need to be analyzed, along the lines
of, say \cite{Ashtekar:2006uz}. One should also
consider application of the Weyl quantization in
cases in which the method from
\cite{Ashtekar:2006wn} cannot be directly
applied. An example for this would be homogenous
but non-isotropic cosmologies.

But arguably the most interesting extension of the
present work would consist in finding a generalization
of the Wigner function to the quantum field theoretic
context of full loop quantum gravity. This seems to be, at the same time, a
very challenging undertaking. In the full theory
wave functions live, roughly speaking, on a certain
inductive limit of products of the Lie group $\textrm{SU}(2)$.
So the two main problems we expect are 1) the definition of a
``good'' Wigner function on $\textrm{SU}(2)$ that 2) interacts well
with taking the inductive limit.
\section*{Acknowledgements}
This work was started during the workshop
\textit{Global Problems in Mathematical
Relativity} at the Isaac Newton Institute for
Mathematical Sciences, Cambridge, and we would
like to thank the institute as well as the
organizers of the workshop, P.T.\ Chrusciel, H.\
Friedrich, and P.\ Tod.

We also thank A.\ Ashtekar, J.\ Lewandowski and
G.A.\ Mena Marug{\'a}n for discussions and M.\
Bojowald for discussions and comments on a draft
of the present paper.

H.S.\ gratefully acknowledges funding for this
work through a Marie Curie Fellowship of the
European Union.

\end{document}